\documentclass[a4paper,10pt]{article}
\usepackage[english]{babel}
\usepackage[T1]{fontenc}
\usepackage[utf8]{inputenc}
\usepackage{amssymb}
\usepackage{physics}
\usepackage{amsmath}
\usepackage{amsthm}
\newtheorem{theorem}{Theorem}
\newtheorem{proposition}{Proposition}
\newtheorem{lemma}{Lemma}
\usepackage{graphicx}
\usepackage[unicode]{hyperref}
\usepackage{tikz}
\usepackage{color}
\usepackage{genyoungtabtikz}
\usepackage[title]{appendix}
\usepackage{indentfirst}
\usepackage{bm}
\usepackage{xhfill}
\usepackage{bbm}
\usepackage{multirow}
\usepackage{array}
\usepackage[centertableaux]{ytableau}
\usepackage{multicol}
\usepackage[shortlabels]{enumitem}
\usepackage{forest}
\usepackage{float}
\usepackage{mathrsfs}

\usetikzlibrary{decorations.pathreplacing}

\usepackage{ytableau}

\hypersetup{
    colorlinks=true,
    linkcolor=blue,
    filecolor=magenta,      
    urlcolor=cyan,
}
 
\usepackage{amsthm}

\theoremstyle{definition}

\usepackage[nottoc]{tocbibind}

\usepackage{authblk}

\title{\textbf{Analytic Properties of the Jost Functions via the Poincaré–Picard Theorem}}

\author[1,2]{\textbf{Yannick Mvondo-She}}
\affil[1]{National Institute of Theoretical and Computational Sciences, Private Bag X1, Matieland, South Africa}
\affil[2]{Department of Physics, University of Pretoria, Private Bag X20, Hatfield 0028, South Africa}

\date{}

\begin{document}

\maketitle

\begin{abstract}

The analytic properties of the Jost functions play a fundamental role in quantum scattering theory, particularly in the study of bound states, resonances, and the analytic continuation of the scattering matrix into the complex energy plane. In the present work, the analyticity of the Jost functions is investigated from the perspective of parameter-dependent ordinary differential equations. Starting from the radial Schrödinger equation for a short-range central potential, a coupled first-order differential system for the coefficient functions associated with the Ricatti--Bessel and Ricatti--Neumann solutions is derived using the method of variation of parameters.

The multivalued character of the Jost functions is shown to originate from the square-root relation between the complex energy and the momentum variable,
\begin{equation*}
k=\sqrt{\frac{2\mu E}{\hbar^2}},
\end{equation*}
which generates a two-sheeted Riemann surface with a branch point at the scattering threshold. By explicitly factorizing the momentum dependence from the Ricatti functions and from the coefficient functions themselves, the original scattering problem is transformed into a differential system whose coefficients are single-valued analytic functions of the complex energy variable.

The transformed system is subsequently rewritten in matrix form, allowing the direct application of the Poincaré--Picard theorem on analytic dependence of solutions of ordinary differential equations on complex parameters. It is thereby proven that the transformed Jost functions are single-valued analytic functions of the energy variable for finite radial distance.

The geometric interpretation of the factorization procedure is further discussed in terms of the topology of the energy Riemann surface and the analytic continuation between its different sheets. The present approach provides a mathematically transparent differential-equation-based framework for the analytic structure of the Jost functions and establishes a direct connection between quantum scattering theory, analytic continuation, and classical theorems on parameter-dependent differential equations.

\end{abstract}

\tableofcontents

\section{Introduction}

The analytic structure of scattering amplitudes occupies a central position in quantum scattering theory. Since the pioneering works of Jost, Newton, Regge, and many others, the study of analytic continuation in the complex energy and momentum planes has provided deep insight into the mathematical foundations of resonances, bound states, and the structure of the scattering matrix \cite{jost1946,newton1982,taylor1972,regge1958}.

Among the fundamental objects of scattering theory are the Jost functions,
\begin{equation}
f_l^{(\mathrm{in})}(E),
\qquad
f_l^{(\mathrm{out})}(E),
\end{equation}
which determine the asymptotic behavior of the radial wave function and encode the analytic properties of the scattering matrix. In particular, the partial-wave \(S\)-matrix may be written as
\begin{equation}
S_l(E)
=
\frac{
f_l^{(\mathrm{out})}(E)
}{
f_l^{(\mathrm{in})}(E)
},
\end{equation}
so that poles of the Jost functions correspond to bound states, virtual states, and resonances \cite{newton1982,taylor1972}. The analytic continuation of the Jost functions into the complex energy plane is therefore of considerable physical and mathematical importance. Resonance phenomena, decay widths, and threshold behavior are all intimately connected with the analytic structure of these functions on the Riemann surface generated by the momentum variable \cite{eden1966,newton1982}.

The origin of this nontrivial analytic structure lies in the relation between the energy and momentum,
\begin{equation}
k
=
\sqrt{\frac{2\mu E}{\hbar^2}},
\label{intro_energy_momentum}
\end{equation}
which introduces a square-root branch point at
\begin{equation}
E=0.
\end{equation}
Consequently, the scattering problem is naturally defined not on the complex energy plane itself, but rather on a two-sheeted Riemann surface \cite{taylor1972,newton1982}. Analytic continuation around the branch point changes the sign of the momentum,
\begin{equation}
k \longrightarrow -k,
\end{equation}
thereby connecting different sheets of the energy surface.

The study of analytic properties of scattering quantities has traditionally been developed through integral equation methods. In particular, the analytic behavior of the Jost solutions is commonly established through the Volterra integral equations satisfied by the radial wave functions \cite{newton1982,taylor1972,faddeev1965}. These approaches have proven extremely successful and form the standard framework of analytic scattering theory. Nevertheless, the analytic dependence of solutions of differential equations on complex parameters provides an alternative and conceptually attractive route to analyticity. Classical results due to Poincaré and Picard establish that solutions of ordinary differential equations depend analytically on external parameters whenever the coefficients of the differential system possess the appropriate analyticity properties \cite{hille1976,coddington1955,ince1956,hartman1964,wasow1965}. Despite the generality and power of these theorems, their application to the analytic structure of the Jost functions has received comparatively little attention.

The purpose of the present work is precisely to investigate the analytic properties of the Jost functions from the perspective of parameter-dependent differential equations. The central idea is to isolate explicitly the multivalued dependence generated by the square-root momentum variable and reformulate the scattering problem in terms of transformed functions whose differential equations possess coefficients analytic in the complex energy variable.

The method employed in the present work is closely related to the factorization approach developed in the context of Jost-matrix representations and resonance theory \cite{rakityansky2012,rakityansky2013}. A modern and comprehensive treatment of the Jost-function formalism, including multichannel scattering theory, resonance phenomena, and analytic continuation techniques, may be found in the monograph \cite{rakityansky2022}. Related aspects of the analytic structure of the Jost matrix, resonance phenomena, and multichannel scattering theory were also investigated in the author's Master's thesis \cite{mvondothesis}. In particular, it was shown that the branching structure associated with the momentum variable can be separated explicitly from the single-valued analytic part of the Jost matrices. The present work revisits this idea from a different perspective and places it within the framework of the Poincaré--Picard theorem.

More specifically, we begin with the radial Schrödinger equation for a short-range central potential and derive a first-order system for the coefficient functions associated with the Ricatti--Bessel and Ricatti--Neumann solutions. We then show that the nonanalytic behavior of the Jost functions originates entirely from explicit odd powers of the momentum variable. By factorizing these momentum-dependent terms, we construct transformed coefficient functions
\begin{equation}
\widetilde{A}_l(E,r),
\qquad
\widetilde{B}_l(E,r),
\end{equation}
whose differential equations involve only coefficients analytic in the complex energy variable. This transformation removes the branching structure associated with the square-root mapping and allows the direct application of the Poincaré--Picard theorem. The analyticity of the transformed Jost functions then follows as a consequence of the analytic dependence of solutions of ordinary differential equations on complex parameters. In this way, the analytic structure of the scattering problem emerges naturally from classical theorems of differential equation theory.

Beyond its mathematical interest, the present approach also provides a transparent geometric interpretation of the energy Riemann surface. The factorization procedure isolates explicitly the nontrivial monodromy associated with analytic continuation around the branch point and separates it from the single-valued analytic part of the scattering solutions.

The principal result established in the present work may be summarized as follows: after explicit factorization of the momentum-dependent branching terms, the transformed Jost functions satisfy a first-order differential system whose coefficients are analytic functions of the complex energy variable. The Poincaré--Picard theorem then implies that the transformed Jost functions are single-valued analytic functions of the energy.

The paper is organized as follows. In Section~2, the radial Schrödinger equation and the Jost functions are introduced, and the first-order differential system governing the coefficient functions is derived using the method of variation of parameters. In Section~3, the analytic structure of the energy variable is examined in detail, with particular emphasis on the square-root branching singularity and the associated two-sheeted Riemann surface. In Section~4, the explicit momentum dependence is factorized from the Ricatti functions and from the coefficient functions themselves, leading to a transformed differential system analytic in the energy variable. In Section~5, the transformed system is rewritten in matrix form and the Poincaré--Picard theorem is applied to establish the analyticity of the transformed Jost functions. In Section~6, the geometric interpretation of the factorization procedure is discussed in terms of the topology of the energy Riemann surface and the analytic continuation between different sheets. Finally, Section~7 summarizes the main results and discusses possible extensions of the present formalism to multichannel scattering problems, long-range interactions, and resonance theory.

The present work therefore provides a differential-equation-based approach to the analytic structure of the Jost functions and establishes a direct connection between scattering theory, parameter-dependent ordinary differential equations, and the geometry of complex analytic continuation.

\section{Radial Schrödinger Equation and Jost Functions}

We consider the radial Schrödinger equation describing the relative motion of two interacting particles in a central potential \(V(r)\). After separation of the angular variables, the radial wave function \(u_l(E,r)\) satisfies the differential equation
\begin{equation}
\left[
\frac{d^2}{dr^2}
+
k^2
-
\frac{l(l+1)}{r^2}
-
V(r)
\right]
u_l(E,r)
=
0,
\label{radial_schrodinger}
\end{equation}
where \(l\) denotes the orbital angular momentum quantum number and
\begin{equation}
k
=
\sqrt{\frac{2\mu E}{\hbar^2}}
\label{momentum_definition}
\end{equation}
is the complex momentum associated with the complex energy \(E\). Here \(\mu\) is the reduced mass of the system. Throughout the present work, the potential \(V(r)\) will be assumed to be short-ranged, namely,
\begin{equation}
\lim_{r\to\infty} V(r)=0,
\end{equation}
sufficiently rapidly so that the asymptotic solutions of Eq.~\eqref{radial_schrodinger} coincide with those of the free radial equation.

\subsection{Free radial solutions}

In the asymptotic region where \(V(r)\to 0\), Eq.~\eqref{radial_schrodinger} reduces to
\begin{equation}
\left[
\frac{d^2}{dr^2}
+
k^2
-
\frac{l(l+1)}{r^2}
\right]
u_l(E,r)
=
0.
\label{free_radial_equation}
\end{equation}
The two linearly independent solutions of Eq.~\eqref{free_radial_equation} are the Ricatti--Bessel and Ricatti--Neumann functions,
\begin{equation}
j_l(kr),
\qquad
y_l(kr).
\end{equation}
From these functions one defines the Ricatti--Hankel functions
\begin{equation}
h_l^{(\pm)}(kr)
=
j_l(kr)
\pm
i y_l(kr),
\label{ricatti_hankel}
\end{equation}
which represent outgoing and incoming spherical waves respectively. Using the asymptotic behavior of the spherical Bessel functions, one obtains
\begin{equation}
h_l^{(\pm)}(kr)
\sim
(\mp i)^{l+1} e^{\pm i kr},
\qquad
r\to\infty.
\label{hankel_asymptotics}
\end{equation}
Equation~\eqref{hankel_asymptotics} shows explicitly that
\begin{equation}
h_l^{(+)}(kr)
\end{equation}
describes an outgoing wave, while
\begin{equation}
h_l^{(-)}(kr)
\end{equation}
describes an incoming wave.

\subsection{Definition of the Jost functions}

The regular solution of Eq.~\eqref{radial_schrodinger} is characterized by the boundary condition
\begin{equation}
u_l(E,0)=0.
\label{regular_boundary}
\end{equation}
Since the Ricatti--Hankel functions form a fundamental set of solutions in the asymptotic region, the regular solution can be written asymptotically as
\begin{equation}
u_l(E,r)
\underset{r\to\infty}{\sim}
f_l^{(\mathrm{in})}(E)
\, h_l^{(-)}(kr)
+
f_l^{(\mathrm{out})}(E)
\, h_l^{(+)}(kr),
\label{jost_asymptotics}
\end{equation}
where
\begin{equation}
f_l^{(\mathrm{in})}(E),
\qquad
f_l^{(\mathrm{out})}(E),
\end{equation}
are respectively called the incoming and outgoing Jost functions. The scattering matrix is then defined by
\begin{equation}
S_l(E)
=
\frac{
f_l^{(\mathrm{out})}(E)
}{
f_l^{(\mathrm{in})}(E)
}.
\label{smatrix_definition}
\end{equation}
The analytic structure of the Jost functions therefore determines the analytic properties of the scattering matrix and consequently the positions of bound states and resonances.

\subsection{Variation of parameters}

In order to derive a differential system for the Jost functions, we use the method of variation of parameters and write the regular solution in the form
\begin{equation}
u_l(E,r)
=
A_l(E,r)\, j_l(kr)
-
B_l(E,r)\, y_l(kr),
\label{variation_ansatz}
\end{equation}
where the functions
\begin{equation}
A_l(E,r),
\qquad
B_l(E,r),
\end{equation}
are now allowed to depend on the radial coordinate.
\\
The representation \eqref{variation_ansatz} is not unique. We therefore impose the auxiliary condition
\begin{equation}
\partial_r A_l(E,r)\, j_l(kr)
-
\partial_r B_l(E,r)\, y_l(kr)
=
0.
\label{auxiliary_condition}
\end{equation}
Differentiating Eq.~\eqref{variation_ansatz} with respect to \(r\) gives
\begin{align}
\partial_r u_l(E,r)
&=
\partial_r A_l(E,r)\, j_l(kr)
+
A_l(E,r)\, \partial_r j_l(kr)
\nonumber\\
&\quad
-
\partial_r B_l(E,r)\, y_l(kr)
-
B_l(E,r)\, \partial_r y_l(kr).
\end{align}
Using the auxiliary condition \eqref{auxiliary_condition}, this simplifies to
\begin{equation}
\partial_r u_l(E,r)
=
A_l(E,r)\, \partial_r j_l(kr)
-
B_l(E,r)\, \partial_r y_l(kr).
\label{first_derivative}
\end{equation}
Differentiating once more yields
\begin{align}
\partial_r^2 u_l(E,r)
&=
\partial_r A_l(E,r)\, \partial_r j_l(kr)
-
\partial_r B_l(E,r)\, \partial_r y_l(kr)
\nonumber\\
&\quad
+
A_l(E,r)\, \partial_r^2 j_l(kr)
-
B_l(E,r)\, \partial_r^2 y_l(kr).
\label{second_derivative}
\end{align}
Substituting Eqs.~\eqref{variation_ansatz} and \eqref{second_derivative} into Eq.~\eqref{radial_schrodinger}, and using the fact that \(j_l(kr)\) and \(y_l(kr)\) satisfy the free equation \eqref{free_radial_equation}, we obtain
\begin{equation}
\partial_r A_l(E,r)\, \partial_r j_l(kr)
-
\partial_r B_l(E,r)\, \partial_r y_l(kr)
=
V(r)
\left[
A_l(E,r)\, j_l(kr)
-
B_l(E,r)\, y_l(kr)
\right].
\label{intermediate_equation}
\end{equation}
Equations \eqref{auxiliary_condition} and \eqref{intermediate_equation} constitute a linear system for
\begin{equation}
\partial_r A_l(E,r),
\qquad
\partial_r B_l(E,r).
\end{equation}
Using the Wronskian relation
\begin{equation}
j_l(kr)\, \partial_r y_l(kr)
-
y_l(kr)\, \partial_r j_l(kr)
=
k,
\label{wronskian_relation}
\end{equation}
one finally obtains
\begin{equation}
\boxed{
\partial_r A_l(E,r)
=
-
\frac{y_l(kr)}{k}
V(r)
\left[
j_l(kr)\, A_l(E,r)
-
y_l(kr)\, B_l(E,r)
\right]
}
\label{A_equation}
\end{equation}
and
\begin{equation}
\boxed{
\partial_r B_l(E,r)
=
-
\frac{j_l(kr)}{k}
V(r)
\left[
j_l(kr)\, A_l(E,r)
-
y_l(kr)\, B_l(E,r)
\right].
}
\label{B_equation}
\end{equation}
The regularity condition at the origin implies
\begin{equation}
A_l(E,0)=1,
\qquad
B_l(E,0)=0.
\label{initial_conditions}
\end{equation}
For short-range potentials, the functions
\begin{equation}
A_l(E,r),
\qquad
B_l(E,r),
\end{equation}
approach finite limits as \(r\to\infty\). Comparing Eq.~\eqref{variation_ansatz} with the asymptotic representation \eqref{jost_asymptotics}, one obtains
\begin{equation}
f_l^{(\mathrm{in})}(E)
=
\frac{1}{2}
\left[
A_l(E,\infty)
+
i B_l(E,\infty)
\right],
\label{incoming_jost}
\end{equation}
and
\begin{equation}
f_l^{(\mathrm{out})}(E)
=
\frac{1}{2}
\left[
A_l(E,\infty)
-
i B_l(E,\infty)
\right].
\label{outgoing_jost}
\end{equation}
Equations \eqref{A_equation} and \eqref{B_equation} constitute the fundamental first-order system from which the analytic properties of the Jost functions will be studied in the following sections.

\section{Analytic Structure of the Energy Variable}

In the previous section, the Jost functions were introduced through the asymptotic representation of the regular solution of the radial Schrödinger equation. The present section is devoted to the analytic structure of these functions as functions of the complex energy variable \(E\). The central mathematical difficulty originates from the relation between the energy and the momentum,
\begin{equation}
k
=
\sqrt{\frac{2\mu E}{\hbar^2}},
\label{energy_momentum_relation}
\end{equation}
which introduces a square-root dependence on the complex energy. As will be shown below, this relation naturally leads to a branching structure in the complex energy plane and consequently to the multivalued character of the Jost functions.

\subsection{Complex energy and multivaluedness}

Let the complex energy be written in polar form,
\begin{equation}
E
=
|E| e^{i\theta},
\qquad
-\infty < \theta < \infty.
\label{polar_energy}
\end{equation}
Substituting Eq.~\eqref{polar_energy} into Eq.~\eqref{energy_momentum_relation} gives
\begin{align}
k
&=
\sqrt{\frac{2\mu}{\hbar^2}}
\sqrt{|E|e^{i\theta}}
\nonumber\\
&=
\sqrt{\frac{2\mu}{\hbar^2}}
|E|^{1/2}
e^{i\theta/2}.
\label{complex_momentum}
\end{align}
Equation~\eqref{complex_momentum} immediately shows that the momentum depends on the half-angle \(\theta/2\). Consequently, when the energy variable performs one complete rotation around the origin,
\begin{equation}
\theta
\longrightarrow
\theta + 2\pi,
\end{equation}
the momentum transforms according to
\begin{align}
k
&\longrightarrow
\sqrt{\frac{2\mu}{\hbar^2}}
|E|^{1/2}
e^{i(\theta+2\pi)/2}
\nonumber\\
&=
\sqrt{\frac{2\mu}{\hbar^2}}
|E|^{1/2}
e^{i\theta/2}
e^{i\pi}
\nonumber\\
&=
-k.
\label{momentum_sign_change}
\end{align}
Therefore, after one complete circuit around the origin in the energy plane, the momentum does not return to its original value. Instead, its sign changes. Only after two complete rotations,
\begin{equation}
\theta
\longrightarrow
\theta + 4\pi,
\end{equation}
does one recover the original value of \(k\). This behavior demonstrates that the momentum is not a single-valued function of the energy. The point
\begin{equation}
E=0
\end{equation}
is therefore a branch point of the mapping
\begin{equation}
E \mapsto k(E).
\end{equation}

\subsection{Two-sheeted Riemann surface}

The multivaluedness of the square-root function can be understood geometrically through the construction of a Riemann surface. To this end, consider again Eq.~\eqref{complex_momentum},
\begin{equation}
k
=
\sqrt{\frac{2\mu}{\hbar^2}}
|E|^{1/2}
e^{i\theta/2}.
\end{equation}
If the argument \(\theta\) is restricted to the interval
\begin{equation}
0 \leq \theta < 2\pi,
\label{first_sheet_interval}
\end{equation}
one obtains one branch of the square-root function. This branch defines the first sheet of the energy Riemann surface. Similarly, restricting the argument to
\begin{equation}
2\pi \leq \theta < 4\pi
\label{second_sheet_interval}
\end{equation}
defines a second branch corresponding to the second sheet. The complete square-root function therefore requires two copies of the complex energy plane connected continuously at the branch point \(E=0\). The resulting surface is called the two-sheeted energy Riemann surface. On the first sheet one has
\begin{equation}
k(E)
=
+\sqrt{\frac{2\mu E}{\hbar^2}},
\label{physical_sheet}
\end{equation}
whereas on the second sheet,
\begin{equation}
k(E)
=
-\sqrt{\frac{2\mu E}{\hbar^2}}.
\label{nonphysical_sheet}
\end{equation}
A convenient way to visualize the construction consists in introducing a branch cut in the energy plane. In scattering theory, the branch cut is usually chosen along the positive real energy axis,
\begin{equation}
E \in [0,\infty),
\end{equation}
so that the two sheets become connected continuously across the cut. The first sheet is commonly referred to as the physical sheet, while the second sheet is called the nonphysical sheet. Resonance poles of the scattering matrix are generally located on the second sheet.

\subsection{Consequences for the Jost functions}

The Jost functions inherit the multivaluedness of the momentum because they depend explicitly on the Ricatti--Hankel functions,
\begin{equation}
h_l^{(\pm)}(kr)
=
j_l(kr)
\pm
i y_l(kr),
\label{ricatti_hankel_recall}
\end{equation}
whose arguments involve the momentum \(k\). To understand the dependence more explicitly, consider the power-series representations of the Ricatti--Bessel and Ricatti--Neumann functions. One has
\begin{equation}
j_l(kr)
=
\sqrt{\pi}
\sum_{n=0}^{\infty}
\frac{
(-1)^n
}{
n!\,
\Gamma\left(l+n+\frac{3}{2}\right)
}
\left(
\frac{kr}{2}
\right)^{2n+l+1},
\label{ricatti_bessel_series}
\end{equation}
and
\begin{equation}
y_l(kr)
=
\sum_{n=0}^{\infty}
c_n
(kr)^{2n-l},
\label{ricatti_neumann_series}
\end{equation}
where \(c_n\) are numerical coefficients independent of the energy. Equation~\eqref{ricatti_bessel_series} shows that \(j_l(kr)\) contains odd powers of the momentum,
\begin{equation}
k^{2n+l+1},
\end{equation}
while Eq.~\eqref{ricatti_neumann_series} contains powers of the form
\begin{equation}
k^{2n-l}.
\end{equation}
Under the transformation
\begin{equation}
k \longrightarrow -k,
\end{equation}
generated by one complete circuit around the branch point, these terms acquire phase changes. Consequently, the Jost functions do not remain invariant under analytic continuation around \(E=0\). The Jost functions are therefore multivalued functions of the complex energy variable.

\subsection{Origin of the branching structure}

The branching structure may now be understood in a precise way. The radial Schrödinger equation itself,
\begin{equation}
\left[
\frac{d^2}{dr^2}
+
k^2
-
\frac{l(l+1)}{r^2}
-
V(r)
\right]
u_l(E,r)
=
0,
\end{equation}
depends only on
\begin{equation}
k^2
=
\frac{2\mu E}{\hbar^2},
\end{equation}
which is single-valued in the energy variable. Hence the differential equation itself possesses no branching singularity. The multivaluedness arises instead from the asymptotic decomposition into incoming and outgoing waves,
\begin{equation}
e^{\pm ikr},
\end{equation}
or equivalently from the Ricatti--Hankel functions. The distinction between incoming and outgoing solutions therefore introduces the square-root branching structure into the theory.

\subsection{Motivation for factorization}

The previous analysis shows that the nonanalytic behavior of the Jost functions originates entirely from the odd powers of the momentum contained in the Ricatti functions. This observation suggests that the multivalued factors may be isolated explicitly. More precisely, one seeks representations of the form
\begin{equation}
j_l(kr)
=
k^{l+1}
\widetilde{j}_l(E,r),
\label{factorized_bessel}
\end{equation}
and
\begin{equation}
y_l(kr)
=
k^{-l}
\widetilde{y}_l(E,r),
\label{factorized_neumann}
\end{equation}
where the functions
\begin{equation}
\widetilde{j}_l(E,r),
\qquad
\widetilde{y}_l(E,r),
\end{equation}
depend only on even powers of the momentum and are therefore single-valued analytic functions of the energy.

The preceding discussion shows that the nonanalyticity of the scattering problem originates entirely from explicit odd powers of the momentum variable. This observation suggests that the branching structure may be isolated through an explicit factorization of the momentum dependence. The purpose of the next section is precisely to implement this factorization procedure and derive a transformed first-order differential system whose coefficients are analytic functions of the complex energy variable.

\section{Factorization of the Momentum Dependence}

In the previous section it was shown that the multivalued character of the Jost functions originates from the square-root dependence
\begin{equation}
k
=
\sqrt{\frac{2\mu E}{\hbar^2}},
\end{equation}
which introduces odd powers of the momentum into the Ricatti--Bessel and Ricatti--Neumann functions.

The purpose of the present section is to isolate explicitly the nonanalytic dependence on the momentum and reformulate the first-order system obtained in Section~2 in terms of functions that are single-valued analytic functions of the complex energy variable \(E\).

\subsection{Factorization of the Ricatti--Bessel function}

We begin with the power-series representation of the Ricatti--Bessel function,
\begin{equation}
j_l(kr)
=
\sqrt{\pi}
\sum_{n=0}^{\infty}
\frac{
(-1)^n
}{
n!\,
\Gamma\left(l+n+\frac{3}{2}\right)
}
\left(
\frac{kr}{2}
\right)^{2n+l+1}.
\label{bessel_series_start}
\end{equation}
Each term in the series contains the factor
\begin{equation}
k^{2n+l+1}.
\end{equation}
Since
\begin{equation}
k^{2n+l+1}
=
k^{l+1}(k^2)^n,
\end{equation}
Eq.~\eqref{bessel_series_start} may be rewritten as
\begin{align}
j_l(kr)
&=
k^{l+1}
\sqrt{\pi}
\sum_{n=0}^{\infty}
\frac{
(-1)^n
}{
n!\,
\Gamma\left(l+n+\frac{3}{2}\right)
}
\left(
\frac{r}{2}
\right)^{2n+l+1}
(k^2)^n.
\label{bessel_rearranged}
\end{align}
Using the relation
\begin{equation}
k^2
=
\frac{2\mu E}{\hbar^2},
\label{k_squared_relation}
\end{equation}
one obtains
\begin{align}
j_l(kr)
&=
k^{l+1}
\sqrt{\pi}
\sum_{n=0}^{\infty}
\frac{
(-1)^n
}{
n!\,
\Gamma\left(l+n+\frac{3}{2}\right)
}
\left(
\frac{r}{2}
\right)^{2n+l+1}
\left(
\frac{2\mu E}{\hbar^2}
\right)^n.
\label{bessel_energy_series}
\end{align}
We now define the factorized Ricatti--Bessel function by
\begin{equation}
\boxed{
j_l(kr)
=
k^{l+1}
\widetilde{j}_l(E,r).
}
\label{factorized_j}
\end{equation}
The reduced function is therefore
\begin{equation}
\widetilde{j}_l(E,r)
=
\sqrt{\pi}
\sum_{n=0}^{\infty}
\frac{
(-1)^n
}{
n!\,
\Gamma\left(l+n+\frac{3}{2}\right)
}
\left(
\frac{r}{2}
\right)^{2n+l+1}
\left(
\frac{2\mu E}{\hbar^2}
\right)^n.
\label{reduced_bessel}
\end{equation}
Equation~\eqref{reduced_bessel} depends only on integer powers of \(E\). Consequently,
\begin{equation}
\widetilde{j}_l(E,r)
\end{equation}
is a single-valued analytic function of the complex energy variable.

\subsection{Factorization of the Ricatti--Neumann function}

We now perform the same analysis for the Ricatti--Neumann function. The series representation may be written schematically as
\begin{equation}
y_l(kr)
=
\sum_{n=0}^{\infty}
c_n
(kr)^{2n-l},
\label{neumann_series_start}
\end{equation}
where the coefficients \(c_n\) are independent of the energy. Each term contains the factor
\begin{equation}
k^{2n-l}.
\end{equation}
Since
\begin{equation}
k^{2n-l}
=
k^{-l}(k^2)^n,
\end{equation}
Eq.~\eqref{neumann_series_start} becomes
\begin{equation}
y_l(kr)
=
k^{-l}
\sum_{n=0}^{\infty}
c_n
r^{2n-l}
(k^2)^n.
\label{neumann_rearranged}
\end{equation}
Using Eq.~\eqref{k_squared_relation}, one obtains
\begin{equation}
y_l(kr)
=
k^{-l}
\sum_{n=0}^{\infty}
c_n
r^{2n-l}
\left(
\frac{2\mu E}{\hbar^2}
\right)^n.
\label{neumann_energy_series}
\end{equation}
We therefore define
\begin{equation}
\boxed{
y_l(kr)
=
k^{-l}
\widetilde{y}_l(E,r),
}
\label{factorized_y}
\end{equation}
with
\begin{equation}
\widetilde{y}_l(E,r)
=
\sum_{n=0}^{\infty}
c_n
r^{2n-l}
\left(
\frac{2\mu E}{\hbar^2}
\right)^n.
\label{reduced_neumann}
\end{equation}
Again, the reduced function depends only on integer powers of the energy variable. Hence,
\begin{equation}
\widetilde{y}_l(E,r)
\end{equation}
is single-valued and analytic in \(E\).

\subsection{Factorization of the coefficient functions}

We now return to the first-order system derived in Section~2,
\begin{equation}
\partial_r A_l(E,r)
=
-
\frac{y_l(kr)}{k}
V(r)
\left[
j_l(kr)\, A_l(E,r)
-
y_l(kr)\, B_l(E,r)
\right],
\label{original_A_equation}
\end{equation}
and
\begin{equation}
\partial_r B_l(E,r)
=
-
\frac{j_l(kr)}{k}
V(r)
\left[
j_l(kr)\, A_l(E,r)
-
y_l(kr)\, B_l(E,r)
\right].
\label{original_B_equation}
\end{equation}
The factorization of the Ricatti--Bessel and Ricatti--Neumann functions
shows that the explicit multivalued dependence on the energy variable
enters through powers of the momentum
\begin{equation}
k=\sqrt{\frac{2\mu E}{\hbar^2}}.
\end{equation}
Our goal is therefore to isolate this dependence explicitly while
preserving the physical normalization of the regular solution at the
origin. To achieve this, we leave the coefficient function \(A_l(E,r)\)
unchanged and introduce a transformed coefficient function
\(\widehat B_l(E,r)\) according to
\begin{equation}
\boxed{
\widehat A_l(E,r)
=
A_l(E,r),
}
\label{factorized_A}
\end{equation}
and
\begin{equation}
\boxed{
\widehat B_l(E,r)
=
k^{-(2l+1)}
B_l(E,r).
}
\label{factorized_B}
\end{equation}
Equivalently,
\begin{equation}
A_l(E,r)
=
\widehat A_l(E,r),
\end{equation}
and
\begin{equation}
B_l(E,r)
=
k^{2l+1}
\widehat B_l(E,r).
\label{inverse_B}
\end{equation}
Since the momentum depends only on the energy variable,
\begin{equation}
\partial_r k=0,
\end{equation}
one obtains
\begin{equation}
\partial_r A_l(E,r)
=
\partial_r \widehat A_l(E,r),
\label{A_derivative_transform}
\end{equation}
and
\begin{equation}
\partial_r B_l(E,r)
=
k^{2l+1}
\partial_r \widehat B_l(E,r).
\label{B_derivative_transform}
\end{equation}

\subsection{Derivation of the transformed system}

We now substitute the factorized expressions into
Eq.~\eqref{original_A_equation}. Using
\begin{equation}
j_l(kr)
=
k^{l+1}
\widetilde j_l(E,r),
\end{equation}
and
\begin{equation}
y_l(kr)
=
k^{-l}
\widetilde y_l(E,r),
\end{equation}
together with Eqs.~\eqref{factorized_A}
and~\eqref{inverse_B}, we obtain
\begin{align}
j_l(kr)A_l(E,r)
&=
k^{l+1}
\widetilde j_l(E,r)
\widehat A_l(E,r),
\\[1ex]
y_l(kr)B_l(E,r)
&=
k^{-l}
\widetilde y_l(E,r)
k^{2l+1}
\widehat B_l(E,r)
\nonumber\\
&=
k^{l+1}
\widetilde y_l(E,r)
\widehat B_l(E,r).
\end{align}
Hence,
\begin{equation}
j_lA_l-y_lB_l
=
k^{l+1}
\left[
\widetilde j_l\widehat A_l
-
\widetilde y_l\widehat B_l
\right].
\label{factorized_combination}
\end{equation}
Furthermore,
\begin{equation}
\frac{y_l(kr)}{k}
=
k^{-(l+1)}
\widetilde y_l(E,r).
\end{equation}
Substituting these expressions into
Eq.~\eqref{original_A_equation},
all powers of \(k\) cancel and one obtains
\begin{equation}
\boxed{
\partial_r \widehat A_l(E,r)
=
-
\widetilde y_l(E,r)
V(r)
\left[
\widetilde j_l(E,r)\widehat A_l(E,r)
-
\widetilde y_l(E,r)\widehat B_l(E,r)
\right].
}
\label{transformed_A_equation}
\end{equation}
We next consider Eq.~\eqref{original_B_equation}.
Since
\begin{equation}
\frac{j_l(kr)}{k}
=
k^l
\widetilde j_l(E,r),
\end{equation}
and
\begin{equation}
\partial_r B_l(E,r)
=
k^{2l+1}
\partial_r \widehat B_l(E,r),
\end{equation}
substitution into Eq.~\eqref{original_B_equation}
gives
\begin{align}
k^{2l+1}
\partial_r \widehat B_l(E,r)
=
-
k^{2l+1}
\widetilde j_l(E,r)
V(r)
\left[
\widetilde j_l(E,r)\widehat A_l(E,r)
-
\widetilde y_l(E,r)\widehat B_l(E,r)
\right].
\end{align}
Dividing by \(k^{2l+1}\), one obtains
\begin{equation}
\boxed{
\partial_r \widehat B_l(E,r)
=
-
\widetilde j_l(E,r)
V(r)
\left[
\widetilde j_l(E,r)\widehat A_l(E,r)
-
\widetilde y_l(E,r)\widehat B_l(E,r)
\right].
}
\label{transformed_B_equation}
\end{equation}

\subsection{Analytic structure of the transformed system}

The transformed differential system
\eqref{transformed_A_equation}--\eqref{transformed_B_equation}
contains only the reduced Ricatti--Bessel and Ricatti--Neumann
functions
\begin{equation}
\widetilde j_l(E,r),
\qquad
\widetilde y_l(E,r),
\end{equation}
which depend on the energy variable only through integer powers of
\(E\). Consequently, all coefficients of the transformed system are
single-valued analytic functions of the complex energy variable. An important feature of the present transformation is that it preserves
the regular normalization at the origin. If
\begin{equation}
A_l(E,0)=1,
\qquad
B_l(E,0)=0,
\end{equation}
then
\begin{equation}
\widehat A_l(E,0)=1,
\qquad
\widehat B_l(E,0)=0.
\end{equation}
The transformed system therefore possesses initial conditions that are
independent of the energy parameter and free from any threshold branch
structure. The factorization procedure has thus isolated the entire multivalued dependence of the scattering problem into explicit powers of the
momentum variable \(k\), while reformulating the differential system in
terms of functions whose coefficients are analytic in the complex
energy plane. This property will allow the direct application of the
Poincar\'e--Picard theorem in the next section.

\section{Analyticity via the Poincaré--Picard Theorem}

In the previous section, the original first-order system for the Jost functions was transformed into a differential system whose coefficients are single-valued analytic functions of the complex energy variable \(E\). The purpose of the present section is to establish rigorously the analyticity of the transformed Jost functions by means of the Poincaré--Picard theorem on analytic dependence of solutions of differential equations on parameters.
The strategy of the proof is the following:

\begin{enumerate}[(i)]
\item rewrite the transformed equations as a matrix differential system;

\item prove that the coefficient matrix is analytic in the energy parameter;

\item apply the Poincaré--Picard theorem;

\item conclude that the transformed Jost functions are analytic functions of the complex energy.
\end{enumerate}

\subsection{Matrix form of the transformed system}

The transformed first-order equations obtained in Section~4 are
\begin{equation}
\partial_r \widehat{A}_l(E,r)
=
-
\widetilde{y}_l(E,r)
V(r)
\left[
\widetilde{j}_l(E,r)\widehat{A}_l(E,r)
-
\widetilde{y}_l(E,r)\widehat{B}_l(E,r)
\right],
\label{transformed_A_recall}
\end{equation}
and
\begin{equation}
\partial_r \widehat{B}_l(E,r)
=
-
\widetilde{j}_l(E,r)
V(r)
\left[
\widetilde{j}_l(E,r)\widehat{A}_l(E,r)
-
\widetilde{y}_l(E,r)\widehat{B}_l(E,r)
\right].
\label{transformed_B_recall}
\end{equation}
Expanding the right-hand sides gives
\begin{align}
\partial_r \widehat{A}_l(E,r)
&=
-
\widetilde{j}_l(E,r)\widetilde{y}_l(E,r)
V(r)
\widehat{A}_l(E,r)
\nonumber\\
&\quad
+
\widetilde{y}_l^{\,2}(E,r)
V(r)
\widehat{B}_l(E,r),
\label{expanded_A_equation}
\end{align}
and
\begin{align}
\partial_r \widehat{B}_l(E,r)
&=
-
\widetilde{j}_l^{\,2}(E,r)
V(r)
\widehat{A}_l(E,r)
\nonumber\\
&\quad
+
\widetilde{j}_l(E,r)\widetilde{y}_l(E,r)
V(r)
\widehat{B}_l(E,r).
\label{expanded_B_equation}
\end{align}
Introducing the vector
\begin{equation}
X(E,r)
=
\begin{pmatrix}
\widehat{A}_l(E,r)
\\
\widehat{B}_l(E,r)
\end{pmatrix},
\label{vector_definition}
\end{equation}
the transformed system may be written in matrix form as
\begin{equation}
\partial_r X(E,r)
=
M(E,r)\,
X(E,r),
\label{matrix_system}
\end{equation}
where
\begin{equation}
M(E,r)
=
V(r)
\begin{pmatrix}
-\widetilde{j}_l(E,r)\widetilde{y}_l(E,r)
&
\widetilde{y}_l^{\,2}(E,r)
\\[2mm]
-\widetilde{j}_l^{\,2}(E,r)
&
\widetilde{j}_l(E,r)\widetilde{y}_l(E,r)
\end{pmatrix}.
\label{matrix_definition}
\end{equation}
The regular solution satisfies
\begin{equation}
A_l(E,0)=1,
\qquad
B_l(E,0)=0.
\end{equation}
Since the transformed variables were defined by
\begin{equation}
\widehat A_l=A_l,
\qquad
\widehat B_l=k^{-(2l+1)}B_l,
\end{equation}
it follows immediately that
\begin{equation}
\boxed{
X(E,0)
=
\begin{pmatrix}
1
\\
0
\end{pmatrix}.
}
\label{matrix_initial_conditions}
\end{equation}
The initial vector is independent of the energy parameter and contains no branch-point dependence.
\subsection{Analyticity of the coefficient matrix}

The applicability of the Poincaré--Picard theorem depends on the analytic properties of the coefficient matrix \(M(E,r)\). We therefore begin by establishing the analyticity of the reduced Ricatti functions.

\begin{lemma}
For every fixed value of \(r\), the functions
\begin{equation}
\widetilde{j}_l(E,r),
\qquad
\widetilde{y}_l(E,r),
\end{equation}
are analytic functions of the complex energy variable \(E\).
\end{lemma}

\begin{proof}

From Section~4, the reduced Ricatti--Bessel function is given by
\begin{equation}
\widetilde{j}_l(E,r)
=
\sqrt{\pi}
\sum_{n=0}^{\infty}
\frac{
(-1)^n
}{
n!\,
\Gamma\left(l+n+\frac{3}{2}\right)
}
\left(
\frac{r}{2}
\right)^{2n+l+1}
\left(
\frac{2\mu E}{\hbar^2}
\right)^n.
\label{reduced_j_recall}
\end{equation}
Equation~\eqref{reduced_j_recall} is a power series in the variable \(E\). Since power series define analytic functions inside their domains of convergence, it follows that
\begin{equation}
\widetilde{j}_l(E,r)
\end{equation}
is analytic in \(E\). Similarly, the reduced Ricatti--Neumann function is
\begin{equation}
\widetilde{y}_l(E,r)
=
\sum_{n=0}^{\infty}
c_n
r^{2n-l}
\left(
\frac{2\mu E}{\hbar^2}
\right)^n,
\label{reduced_y_recall}
\end{equation}
which is again a power series in the energy variable. Therefore,
\begin{equation}
\widetilde{y}_l(E,r)
\end{equation}
is analytic in \(E\).

\end{proof}
\noindent
We may now establish the analyticity of the coefficient matrix.

\begin{proposition}
For every fixed value of \(r\), the matrix
\begin{equation}
M(E,r)
\end{equation}
is analytic in the complex energy variable \(E\).
\end{proposition}

\begin{proof}

The entries of the matrix \(M(E,r)\) are finite products of the functions
\begin{equation}
\widetilde{j}_l(E,r),
\qquad
\widetilde{y}_l(E,r),
\end{equation}
together with the potential \(V(r)\). Since \(V(r)\) is independent of the energy parameter and the reduced Ricatti functions are analytic in \(E\), the entries
\begin{equation}
\widetilde{j}_l(E,r)\widetilde{y}_l(E,r),
\qquad
\widetilde{j}_l^2(E,r),
\qquad
\widetilde{y}_l^2(E,r),
\end{equation}
are also analytic in \(E\), because sums and products of analytic functions remain analytic. Consequently, all entries of the matrix \(M(E,r)\) are analytic functions of the complex energy variable. Therefore the matrix itself is analytic in \(E\).

\end{proof}

\subsection{The Poincaré--Picard theorem}
The following theorem is a standard result in the theory of ordinary differential equations with analytic parameter dependence \cite{coddington1955,hartman1964,wasow1965}.

\begin{theorem}[Poincaré--Picard]
Consider the differential system
\begin{equation}
\frac{dY}{dr}
=
F(r,Y,\lambda),
\end{equation}
where \(F\) is analytic in the complex parameter \(\lambda\) and continuous in the remaining variables. If the initial data are analytic functions of \(\lambda\), then the corresponding solution depends analytically on \(\lambda\).
\end{theorem}
\noindent
The theorem expresses the analytic dependence of the solutions of differential equations on external complex parameters. In the present work, the role of the parameter \(\lambda\) is played by the complex energy variable \(E\).

\subsection{Application to the transformed Jost system}

We now verify explicitly that the transformed system satisfies the hypotheses of the Poincar\'e--Picard theorem.

\begin{enumerate}[(i)]

\item The coefficient matrix \(M(E,r)\) is analytic in the energy variable \(E\).

\item The initial condition
\[
X(E,0)
=
\begin{pmatrix}
1\\
0
\end{pmatrix}
\]
is independent of the energy parameter and therefore analytic in \(E\).

\item The short-range potential \(V(r)\) is assumed sufficiently regular to guarantee existence and uniqueness of solutions.

\end{enumerate}
Hence all assumptions of the Poincar\'e--Picard theorem are satisfied.
The theorem therefore guarantees that the solution vector
\begin{equation}
X(E,r)
=
\begin{pmatrix}
\widehat{A}_l(E,r)
\\
\widehat{B}_l(E,r)
\end{pmatrix}
\end{equation}
depends analytically on the complex energy variable \(E\) for every finite value of the radial coordinate \(r\). We may now state the main theorem of the present work.

\begin{theorem}
Let \(V(r)\) be a short-range potential for which the transformed first-order system admits a unique solution. Then, for every finite value of the radial coordinate \(r\), the transformed coefficient functions
\begin{equation}
\widehat A_l(E,r),
\qquad
\widehat B_l(E,r),
\end{equation}
are single-valued analytic functions of the complex energy variable \(E\).
\label{main_theorem}
\end{theorem}

\begin{proof}

From Proposition~1, the coefficient matrix
\begin{equation}
M(E,r)
\end{equation}
is analytic in the complex energy variable. Furthermore, the transformed system satisfies the energy-independent initial condition
\begin{equation}
X(E,0)
=
\begin{pmatrix}
1\\
0
\end{pmatrix},
\end{equation}
which is analytic in \(E\). The transformed differential system therefore satisfies all hypotheses of the Poincaré--Picard theorem. Consequently, the solution vector
\begin{equation}
X(E,r)
=
\begin{pmatrix}
\widehat A_l(E,r)
\\
\widehat B_l(E,r)
\end{pmatrix}
\end{equation}
depends analytically on the complex energy parameter \(E\). Hence both
\begin{equation}
\widehat A_l(E,r)
\qquad
\text{and}
\qquad
\widehat B_l(E,r)
\end{equation}
are analytic functions of \(E\). Since the explicit square-root dependence associated with the momentum variable has been removed from the transformed differential system, these functions are single-valued on the complex energy plane.

\end{proof}

\subsection{Interpretation of the result}

The significance of Theorem~\ref{main_theorem} lies in the fact that the original branching structure associated with the momentum variable has been completely isolated through factorization. Indeed, the original Jost functions inherit the multivaluedness of the square-root mapping
\begin{equation}
k
=
\sqrt{\frac{2\mu E}{\hbar^2}},
\end{equation}
whereas the transformed quantities
\begin{equation}
\widetilde{A}_l(E,r),
\qquad
\widetilde{B}_l(E,r),
\end{equation}
depend only on integer powers of the energy variable. The transformed differential system therefore possesses coefficients analytic in \(E\), which allows the direct application of the Poincaré--Picard theorem. The analyticity established in the present section forms the mathematical foundation for the study of resonances, poles of the scattering matrix, and analytic continuation onto different sheets of the energy Riemann surface.

\section{Discussion of the Riemann Surface}

In Section~5 it was shown that the transformed coefficient functions
\begin{equation}
\widehat A_l(E,r),
\qquad
\widehat B_l(E,r),
\end{equation}
are single-valued analytic functions of the complex energy variable \(E\).
The purpose of the present section is to interpret this result
geometrically in terms of the Riemann surface generated by the
square-root relation between the momentum and the energy. The analysis clarifies the origin of the multivaluedness of the physical Jost functions and shows that the entire branching structure of the scattering problem is carried by explicit powers of the momentum variable,
\begin{equation}
k
=
\sqrt{\frac{2\mu E}{\hbar^2}},
\end{equation}
while the transformed coefficient functions remain analytic on the
complex energy plane.

\subsection{The square-root mapping}

The analytic structure of the scattering problem is governed fundamentally by the relation
\begin{equation}
k
=
\sqrt{\frac{2\mu E}{\hbar^2}}.
\label{riemann_k_definition}
\end{equation}
The square-root function is intrinsically multivalued in the complex plane. To see this explicitly, let the complex energy be represented in polar form,
\begin{equation}
E
=
|E|e^{i\theta}.
\label{riemann_energy_polar}
\end{equation}
Substituting Eq.~\eqref{riemann_energy_polar} into Eq.~\eqref{riemann_k_definition} yields
\begin{align}
k
&=
\sqrt{\frac{2\mu}{\hbar^2}}
\sqrt{|E|e^{i\theta}}
\nonumber\\
&=
\sqrt{\frac{2\mu}{\hbar^2}}
|E|^{1/2}
e^{i\theta/2}.
\label{riemann_k_polar}
\end{align}
The momentum therefore depends on the half-angle \(\theta/2\). Consequently, when the energy performs one complete rotation around the origin,
\begin{equation}
\theta
\longrightarrow
\theta+2\pi,
\end{equation}
one obtains
\begin{align}
k
&\longrightarrow
\sqrt{\frac{2\mu}{\hbar^2}}
|E|^{1/2}
e^{i(\theta+2\pi)/2}
\nonumber\\
&=
\sqrt{\frac{2\mu}{\hbar^2}}
|E|^{1/2}
e^{i\theta/2}
e^{i\pi}
\nonumber\\
&=
-k.
\label{sign_change_recall}
\end{align}
The momentum therefore changes sign after one complete circuit around the origin in the energy plane. Only after two complete rotations,
\begin{equation}
\theta
\longrightarrow
\theta+4\pi,
\end{equation}
does one recover the original value of the momentum. This demonstrates that the point
\begin{equation}
E=0
\end{equation}
is a branch point of the square-root mapping.

\subsection{Construction of the energy Riemann surface}

The multivaluedness of the momentum can be resolved geometrically by introducing a Riemann surface. We first define two branches of the square-root function.

\subsubsection*{First sheet}

On the first sheet, the argument of the energy variable is restricted to
\begin{equation}
0 \leq \theta < 2\pi.
\label{first_sheet_theta}
\end{equation}
The momentum is then defined by
\begin{equation}
k_1(E)
=
\sqrt{\frac{2\mu}{\hbar^2}}
|E|^{1/2}
e^{i\theta/2}.
\label{first_sheet_k}
\end{equation}
This branch corresponds to the physical sheet of the scattering problem.

\subsubsection*{Second sheet}

On the second sheet, the argument satisfies
\begin{equation}
2\pi \leq \theta < 4\pi.
\label{second_sheet_theta}
\end{equation}
The corresponding momentum is
\begin{equation}
k_2(E)
=
\sqrt{\frac{2\mu}{\hbar^2}}
|E|^{1/2}
e^{i\theta/2}.
\label{second_sheet_k}
\end{equation}
Using
\begin{equation}
e^{i(\theta+2\pi)/2}
=
-
e^{i\theta/2},
\end{equation}
one obtains
\begin{equation}
k_2(E)
=
-k_1(E).
\label{sheet_relation}
\end{equation}
The two branches are therefore connected continuously at the branch point \(E=0\). The complete square-root function is thus single-valued only on a two-sheeted Riemann surface.

\subsection{Branch cuts and analytic continuation}

In practical applications, one introduces a branch cut in the complex energy plane in order to define a single branch of the square-root function. In scattering theory, the branch cut is usually chosen along the positive real energy axis,
\begin{equation}
E \in [0,\infty).
\label{branch_cut}
\end{equation}
The two sheets are then connected continuously across the cut. Analytic continuation across the branch cut moves the solution from one sheet to the other. If the energy approaches the cut from above,
\begin{equation}
E
=
|E|e^{i0},
\end{equation}
one has
\begin{equation}
k
=
+\sqrt{\frac{2\mu E}{\hbar^2}}.
\end{equation}
Approaching the same point from below gives
\begin{equation}
E
=
|E|e^{i2\pi},
\end{equation}
and therefore
\begin{equation}
k
=
-\sqrt{\frac{2\mu E}{\hbar^2}}.
\end{equation}
Thus the sign of the momentum distinguishes the two sheets of the Riemann surface.

\subsection{Riemann surface and the Jost functions}

The original Jost functions inherit the multivaluedness of the momentum because they depend explicitly on the Ricatti--Hankel functions
\begin{equation}
h_l^{(\pm)}(kr)
=
j_l(kr)
\pm
i y_l(kr).
\label{riemann_hankel}
\end{equation}
Since the arguments of the Ricatti functions contain the momentum \(k\), analytic continuation around the branch point changes the sign of the momentum and consequently modifies the Jost functions. To see this explicitly, consider the reconstructed Jost functions
\eqref{factorized_jost_in} and \eqref{factorized_jost_out}.
Under analytic continuation around the threshold branch point,
\begin{equation}
k
\longrightarrow
-k.
\end{equation}
Since
\begin{equation}
k^{2l+1}
\longrightarrow
(-1)^{2l+1}
k^{2l+1}
=
-k^{2l+1},
\end{equation}
the Jost functions change sheets according to
\begin{align}
f_l^{(\mathrm{in})}
&\longrightarrow
\frac12
\left[
\widehat A_l
-
i\,k^{2l+1}\widehat B_l
\right],
\\
f_l^{(\mathrm{out})}
&\longrightarrow
\frac12
\left[
\widehat A_l
+
i\,k^{2l+1}\widehat B_l
\right].
\end{align}
The analytic functions
\begin{equation}
\widehat A_l(E,\infty),
\qquad
\widehat B_l(E,\infty),
\end{equation}
remain unchanged during the continuation. The sheet structure is
therefore generated entirely by the odd power \(k^{2l+1}\).

\subsection{Reconstruction of the Jost functions}

The results of Section~5 establish the analyticity of the transformed
coefficient functions
\begin{equation}
\widehat A_l(E,r),
\qquad
\widehat B_l(E,r).
\end{equation}
To recover the physical scattering quantities, one must reconstruct the
original coefficient functions and subsequently the Jost functions. From the transformation introduced in Section~4,
\begin{equation}
A_l(E,r)
=
\widehat A_l(E,r),
\end{equation}
and
\begin{equation}
B_l(E,r)
=
k^{2l+1}
\widehat B_l(E,r).
\label{reconstruction_B}
\end{equation}
Taking the limit \(r\to\infty\), one obtains
\begin{equation}
A_l(E,\infty)
=
\widehat A_l(E,\infty),
\end{equation}
and
\begin{equation}
B_l(E,\infty)
=
k^{2l+1}
\widehat B_l(E,\infty).
\end{equation}
The incoming and outgoing Jost functions are defined by
\begin{equation}
f_l^{(\mathrm{in})}(E)
=
\frac12
\left[
A_l(E,\infty)
+
i\,B_l(E,\infty)
\right],
\end{equation}
and
\begin{equation}
f_l^{(\mathrm{out})}(E)
=
\frac12
\left[
A_l(E,\infty)
-
i\,B_l(E,\infty)
\right].
\end{equation}
Substituting the transformed variables yields
\begin{equation}
\boxed{
f_l^{(\mathrm{in})}(E)
=
\frac12
\left[
\widehat A_l(E,\infty)
+
i\,k^{2l+1}
\widehat B_l(E,\infty)
\right],
}
\label{factorized_jost_in}
\end{equation}
and
\begin{equation}
\boxed{
f_l^{(\mathrm{out})}(E)
=
\frac12
\left[
\widehat A_l(E,\infty)
-
i\,k^{2l+1}
\widehat B_l(E,\infty)
\right].
}
\label{factorized_jost_out}
\end{equation}
Since
\begin{equation}
\widehat A_l(E,\infty),
\qquad
\widehat B_l(E,\infty)
\end{equation}
are analytic functions of \(E\), the only remaining source of
multivaluedness is the explicit factor
\begin{equation}
k^{2l+1}.
\end{equation}
The branching structure of the Jost functions is therefore completely
isolated and represented by a single explicit momentum factor.

\subsection{Geometric interpretation of the factorization}

The factorization introduced in Section~4 admits a natural geometric interpretation. The original Jost functions are multivalued because they are defined on the two-sheeted Riemann surface associated with the square-root mapping
\begin{equation}
E \mapsto \sqrt{E}.
\end{equation}
The factors
\begin{equation}
k^{l+1},
\qquad
k^{-l},
\end{equation}
carry all information about the branching structure. After these factors are removed, the remaining functions
\begin{equation}
\widehat{A}_l(E,r),
\qquad
\widehat{B}_l(E,r),
\end{equation}
become single-valued analytic functions on the complex energy plane. The factorization procedure therefore transforms the original multivalued scattering problem into a differential system analytic in the energy variable. In this sense, the factorization removes the branching structure generated by the momentum variable and trivializes the analytic continuation around the branch point.

\subsection{Physical interpretation}

The geometry of the energy Riemann surface is closely connected with the analytic structure of the scattering matrix.
\\
On the physical sheet:

\begin{itemize}
\item bound states correspond to poles located on the negative real energy axis;

\item scattering states correspond to positive real energies.
\end{itemize}
\noindent
On the second sheet:

\begin{itemize}
\item resonance poles appear at complex energies;

\item the imaginary part of the energy determines the decay width of the resonance.
\end{itemize}
\noindent
The results obtained in the present work therefore establish a precise
separation between the analytic and topological aspects of the
scattering problem. The transformed coefficient functions
\begin{equation}
\widehat A_l(E,r),
\qquad
\widehat B_l(E,r),
\end{equation}
are single-valued analytic functions of the complex energy variable,
as guaranteed by the Poincar\'e--Picard theorem. The entire
multivalued structure of the physical Jost functions is carried by the
explicit factor \(k^{2l+1}\), which reflects the topology of the
two-sheeted Riemann surface associated with the threshold branch point
at \(E=0\). The factorization procedure therefore provides a mathematically
transparent description of the analytic continuation of the Jost
functions and of the geometric origin of resonances, bound states, and
sheet transitions in quantum scattering theory.

\section{Conclusion}

In the present work, the analytic structure of the Jost functions has been investigated through the framework of differential equations with analytic parameter dependence. The central objective was to establish the analyticity of suitably transformed coefficient functions associated with the Jost representation by means of the Poincaré--Picard theorem.

We began by considering the radial Schrödinger equation for a short-range central potential and introduced the incoming and outgoing Jost functions through the asymptotic decomposition of the regular radial solution. Using the method of variation of parameters, the scattering problem was reformulated as a coupled first-order differential system for the coefficient functions
\begin{equation}
A_l(E,r),
\qquad
B_l(E,r).
\end{equation}

The analytic structure of the energy variable was then examined in detail. It was shown that the relation
\begin{equation}
k
=
\sqrt{\frac{2\mu E}{\hbar^2}}
\end{equation}
introduces a square-root branching singularity at
\begin{equation}
E=0,
\end{equation}
thereby generating a two-sheeted Riemann surface for the complex energy variable. The multivaluedness of the original Jost functions was traced explicitly to the odd powers of the momentum contained in the Ricatti--Bessel and Ricatti--Neumann functions.

To isolate the branching structure, the Ricatti functions were factorized into explicit momentum-dependent factors and reduced functions depending only on even powers of the momentum. This factorization naturally induced transformed coefficient functions
\begin{equation}
\widehat{A}_l(E,r),
\qquad
\widehat{B}_l(E,r),
\end{equation}
for which the corresponding first-order differential system possesses coefficients analytic in the complex energy variable.

The transformed system was then rewritten in matrix form, making it possible to apply the Poincaré--Picard theorem on analytic dependence of solutions on complex parameters. Since the coefficient matrix was shown to be analytic in the energy variable and the initial conditions were independent of the energy, the theorem implied directly that the transformed coefficient functions are single-valued analytic functions of the complex energy variable.

The geometric interpretation of these results was subsequently discussed in terms of the Riemann surface associated with the square-root momentum mapping. In particular, it was shown that the factorization procedure isolates the nontrivial branching structure carried by the momentum variable and separates it from the analytic component of the scattering problem, transforming the original multivalued scattering problem into a differential system analytic on the complex energy plane.

A key outcome of the analysis is that the multivalued character of the
physical Jost functions can be represented explicitly through the factor
\begin{equation}
k^{2l+1},
\end{equation}
while the remaining coefficient functions are analytic in the complex
energy variable. The branch-point structure associated with the
scattering threshold is therefore separated from the analytic dynamics
described by the transformed differential system.

The present approach provides an alternative perspective on the analytic theory of scattering. Rather than deriving analyticity primarily from integral equations or from direct asymptotic analysis, the method emphasizes the analytic dependence of differential equations on parameters. In this way, the analyticity of the transformed coefficient functions emerges naturally as a consequence of classical existence and analyticity theorems for ordinary differential equations.

Several possible extensions of the present work may be considered. First, the formalism may be generalized to multichannel scattering problems, where the energy Riemann surface possesses a considerably richer topology due to the presence of multiple threshold branch points. Second, the method may be extended to Coulomb-type long-range interactions, where logarithmic branching structures appear in addition to the square-root singularities discussed in the present work. Third, the relation between the present factorization procedure and monodromy theory may provide additional geometric insight into the analytic continuation properties of scattering amplitudes across different sheets of the energy Riemann surface. Finally, the analytic framework developed here may be useful in the numerical study of resonances and pole trajectories, particularly in methods based on analytic continuation into the complex energy plane.

In conclusion, the factorization of the momentum dependence together
with the Poincar\'e--Picard theorem provides a mathematically
transparent framework for the analytic structure of the Jost
representation. The transformed coefficient functions are shown to be
single-valued analytic functions of the complex energy variable, while
the entire branching structure of the physical Jost functions is
isolated into an explicit factor \(k^{2l+1}\) reflecting the topology
of the two-sheeted energy Riemann surface. This establishes a direct
connection between quantum scattering theory, analytic continuation,
Riemann-surface geometry, and the classical theory of parameter-dependent
ordinary differential equations.

\paragraph{Acknowledgements} It is a pleasure for the author to thank Sergei Rakitianski for introducing him to the Jost functions. The author acknowledges financial support from the Department of Physics at the University of Pretoria.

\clearpage

\bibliographystyle{utphys}
\bibliography{sample}

\end{document}